\documentclass[12pt]{article}

\usepackage{authblk}
\usepackage{amssymb}
\usepackage{amsmath}
\usepackage{mathrsfs}
\usepackage{amsfonts}
\usepackage{braket}
\usepackage{color}
\usepackage{graphicx}
\usepackage{tikz}
\usetikzlibrary{shapes,snakes}
\usepackage{pgfplots}
\pgfplotsset{compat=1.9} 
\newtheorem{thm}{Theorem}

\newtheorem{lem}[thm]{Lemma}

\newtheorem{defn}[thm]{Definition}
\newcommand{\norm}[1]{\left\Vert#1\right\Vert}
\newcommand{\abs}[1]{\left\vert#1\right\vert}

\def\XXint#1#2#3{{\setbox0=\hbox{$#1{#2#3}{\int}$ }
\vcenter{\hbox{$#2#3$ }}\kern-.6125\wd0}}

\newcounter{lastnote}

\title{Covariant Quantum White Noise from Light-like Quantum Fields}
\author{Radhakrishnan Balu}
\affil{Army Research Laboratory Adelphi, MD, 21005-5069, USA \\
       radhakrishnan.balu.civ@mail.mil
      }  
\affil{Department of Mathematics \&
	Norbert Wiener Center for Harmonic Analysis and Applications, 
	University of Maryland,
	College Park, MD 20742.\\
	rbalu@umd.edu}          
\date{Received: date / Accepted: \today}
\begin{document}
\maketitle

\begin{abstract}
 We derive covariant Weyl operators for light-like fields, with the massless Weyl fermion as an illustrative example, in such a way that they correspond to quantum white noises in vacuum state of a symmetric Fock space. First, we build a representation of a light-like little group in terms of Weyl operators. We then use this construction to induce a representation of Poincar$\grave{e}$ group to construct relativistic quantum white noises from the fields via Mackey{'}s systems of imprimitivity (SI) machinery. Our construction proceeds by fashioning the fermionic processes on a symmetric Fock space using reflection and identifying the corresponding processes on the isomorphic white noise space. 
\end{abstract}

\section{Introduction}
\label{intro}
Quantum white noise \cite {Hida1977} is an essential tool to describe open quantum systems such as quantum optics based processes \cite {Accardi1990}, \cite {Belavkin2005},  \cite {Tezak2012}, \cite {KP1992} that is part of the quantum stochastic calculus (QSC) framework. A relativistically consistent version of QSC are also available in the literature \cite {Applebaum1995},\cite {Frigerio1989} as relativistic effects are an important consideration in fermionic systems of topological quantum systems \cite {Witten2016}. For in-depth discussions on quantum white noise we recommend the works of Accardi \cite {Accardi2002} and Obata \cite {Obata1994}. Here, we present a systematic formulation of covariant quantum white noise for all fundamental particles by extending the theory for relativistic quantum particles, that is based on systems of imprimitivity, to field theoretic context. Systems of imprimitivity \cite {WignerLz1939}, \cite{Mackey1963}, \cite {Varadarajan1985} is a powerful characterization of dynamical systems, when the configuration space of a quantum system is described by a group, from which infinitesimal forms in terms of differential equations ($Shr\ddot{o}dinger$, Heisenberg, and Dirac etc), and the canonical commutation relations can be derived \cite {Wigner1949} \cite {Rad2018} and \cite {Rad2019}. In our earlier work we have constructed covariant Quantum Fields via Lorentz Group Representation of Weyl Operators \cite {Rad2020}. Here we continue this program by building the Fock spaces for massless Weyl fermions and then the covariant white noises.

Let us review the notion of SI and an important theorem by Mackey that characterizes such systems in terms of induced representations, key notions in Clifford algebras, spinor fields, and Schwartz spaces \cite {Varadarajan1985} before discussing our main result.

\section {Little groups (stabilizer subgroups)}
Some of the different systems of imprimitivity that live on the orbits of the stabilizer subgroups are described below. It is good to keep in mind the picture that SI is an irreducible unitary representation of Poincar$\grave{e}$ group $\mathscr{P}^+$ induced from the representation of a subgroup such as $SO_3$, which is a subgroup of homogeneous Lorentz, as $(U_m(g)\psi)(k) = e^{i\{k,g\}}\psi(R^{-1}_mk)$ where g belongs to the $\mathscr{R}^4$ portion of the Poincar$\grave{e}$ group, m is a member of the rotation group, and the duality between the configuration space $\mathbb{R}^4$ and the momentum space $\mathbb{P}^4$ is expressed using the character the irreducible representation of the group $\mathbb{R}^4$ as:
\begin {align} \label {eq: poissonBr}
\{k,g\} &= k_0 g_0 - k_1 g_1 - k_2 g_2 - k_3 g_3, p \in \mathbb{P}^4. \\
\hat{p}:x &\rightarrow e^{i\{k,g\}}. \\
\{Lx, Lp \} &= \{ x, p \}. \\\
\hat{p}(L^{-1}x) &= \hat{Lp}(x).
\end {align}
In the above L is a matrix representation of Lorentz group acting on $\mathbb{R}^4$ as well as $\mathbb{P}^4$ and it is easy to see that $p \rightarrow Lp$ is the adjoint of L action on $\mathbb{P}^4$. The $\mathbb{R}^4$ space is called the configuration space and the dual $\mathbb{P}^4$ is the momentum space of a relativistic quantum particle.

The stabilizer subgroup of the Poincar$\grave{e}$ group $\mathscr{P}^+$ (Light-like particle) can be described as follows: \cite {Kim1991}: There is no frame in which the relativistic quantum particle is at rest but the frame where the momentum is proportional to $(1,0,0,1)$ has the stabilizer subgroup with elements of the form with generators in terms of Pauli matrices as $\sigma_3, N_1 = \begin{bmatrix} 0 & 1 \\ 0 & 0 \end{bmatrix}, N_1 = \begin{bmatrix} 0 & i \\ 0 & 0 \end{bmatrix} $, that is a rotation around the Z axis and the boost $\Lambda_p$ in the spatial direction \cite {Kim1991}. This subgroup is isomorphic to the two dimensional Euclidean group $E_2$, a Weyl representation of which will help us to build the desired representation of the Poincar$\grave{e}$ group.
\
\section {Fiber bundle representation of relativistic quantum particles}

The states of a freely evolving relativistic quantum particles are described by unitary irreducible representations of Poincar$\grave{e}$, which is a semidirect product of two groups, group that has a geometric interpretation in terms of fiber bundles. 

\begin {defn}
Let A and H be two groups and for each $h\in{H}$ let $t_{h}:a\rightarrow{h[a]}$ be an automorphism (defined below) of the group A. Further, we assume that $h\rightarrow{t_h}$ is a homomorphism of H into the group of automorphisms of A so that
\begin {align} \label{semidirectEq}
h[a] &= hah^{-1}, \forall{a\in{A}}. \\
h &= e_H, \text{  the identity element of H}. \\
t_{{h_1}{h_2}} &= t_{h_1}t_{h_2}.
\end {align}
Now, $G=H\rtimes{A}$ is a group with the multiplication rule of $(h_1,a_1)(h_2,a_2) = (h_1{h_2},a_1{t_{h_1}}[a_2])$. The identity element is $(e_H,e_A)$ and the inverse is given by $(h,a)^{-1} = (h^{-1},h^{-1}[a^{-1}]$. When H is the homogeneous Lorentz group and A is $\mathbb{R}^4$ we get the Poincar$\grave{e}$ group $\mathscr{P}=H\rtimes{A}$ via this construction. The covering group of inhomogeneous Lorentz is also a semidirect product as $\mathscr{P}^* = H^*\rtimes \mathbb{R}^4$ and as every irreducible projective representation of $\mathscr{P}$ is uniquely induced from a representation of $\mathscr{P}^*$ we will work with the covering group, whose orbits in momentum space are smooth, in the following.
\end {defn}
We will need the following lemma for our discussions on constructing induced representations using characters of an abelian group as in the case of equation \eqref {eq: poissonBr}.
\begin {lem} (Lemma 6.12 \cite {Varadarajan1985})
Let $h \in H$. Then, $\forall x \in \hat{A}$ where $\hat{A}$ is the set of characters of the group A (which in our case is $\mathbb{R}^4$), there exists one and only $y \in \hat{A}$ such that $y(a) = x(h^{-1}[a]), \forall a \in A$. If we write y = h[x], then $h,x \rightarrow h[x]$ is continuous from $H \times \hat{A}$ into $\hat{A}$ and $\hat{A}$ becomes a H-space. Here, y can be thought as the adjoint for action of H on $\hat{A}$ and the map $\hat{p}$ in equation\eqref {eq: poissonBr} is such an example that is of interest to our constructions. In essence, we have Fourier analysis when restricted to the abelian grou p A of the semidirect product.
\end {lem}
\section {Clifford Algebras and Spinor Fields}
The states of a relativistic particle are usually described in spinor fields in space-time that are based on Clifford algebras and so let use recollect the relevant notions here. Let us denote by $\rho$ the unique, up to an equivalence, finite dimensional irreducible representation of a simple associative Clifford algebra $\mathscr{C}$ of even dimension $n \ge 4$ with the corresponding nonsingular symmetric bilinear form B. Let $\Sigma$ be the the vector space in which $\rho$ acts and $G_B$ be the orthogonal group of $\Sigma$. Then, we can consider the automorphism $x \rightarrow x^L$ of $\mathscr{C}$ induced by $L \in G_B$ which may be thought of as an irreducible representation of $\mathscr{C}$ in $\Sigma$. As a consequence we have an invertible linear transformation $S_1$ on $\Sigma$ such that 
\begin {equation} \label {ref:rhorep}
\rho(x^L) = S_1 \rho(x) S_1^{-1}.
\end {equation}
This can be seen by looking at the action on the vector space as $\rho(\pi(Lv)) = S_1 \rho(\pi(v)) S_1^{-1}, v \in \Sigma$. We may choose to have $S_1$ unit determinant and by Schur{'}s lemma it is up to a constant multiplication factor. That means for $L_1, L_2, L_1 L_2 \in G_B$ we select $S_1, S_2, S_0$ such that the above equation \ref {ref:rhorep} is satisfied giving us the relation
\begin {equation} \label {ref:unidet}
S_0 = \sigma S_1 S_2 \text{ with }\abs{\sigma} = 1.
\end {equation}
Now, we can proceed to lift this representation to that of the covering group by outlining details of Theorem 9.5 \cite {Varadarajan1985}. Let us denote by $G_\Sigma$ the set of all invertible linear transformations of $\Sigma$ with unit determinant that forms an lcsc group. Then we can have a closed subset of $(L, S_1) \in G_B \times G_\Sigma$ that obeys $\rho(\pi(Lv))S_1 = S_1 \rho(\pi(v))$ and the map $(L,S_1)\rightarrow L$ is continuous on to $G_B$.This guarantees a Borel map $L \in G_B \rightarrow S_1(L) \in G_\Sigma$ such that equation \ref {ref:rhorep} is satisfied for $x \in \mathscr{C}$. As equation \ref {ref:unidet} implies that $L \rightarrow S_1(L)$ is a projective (but non unitary) representation we have the relation $S_1(L_1 L_2) = m(L, L_1)S_1(L)1)S_1(L_2), \forall L_1,L_2 \in G_B$. By associative law $L(L_1 L_2) = (LL_1)L_2$ we get the map m is a multiplier of the semisimple group $G_B)$ the representation $S_1$ can be lifted to the covering group $G_B^*$ as S{'}. Now, S{'} is uniquely determined by $S_1$ as the covering group does not have any one dimensional representation other than the trivial one. We thus get the spinor representation of $G_B^*$ as the map $S{'}(l \rightarrow S{'}(l)), l \in G_B^*$ satisfying the relation $\rho(\pi(L^l v)) = S{'}(l)\rho(\pi(v))S{'}(l)^{-1}$.

Let $V^0$ be a real vector space of even dimension $n \ge 4$ and $B^0$ a nonsingular bilinear symmetric form on $V^0 \times V^0$ and the corresponding complexifications are V and B respectively. Now, the group of invertible endomorphisms of $V^0$ with unit determinant and leave $B^0$ invariant is semisimple and we denote by $G_B^0$ the connected component of identity of this group. By identifying each endomorphism of $V^0$ with its complexification we have $G_B^0 \subseteq G_B$. Then, we have a unique representation of the spinors as $\rho(\pi(L^l v)) = S{'}(l)\rho(\pi(v))S{'}(l)^{-1}, \forall v \in V^0, l \in G_B^*$. We can now use Dirac matrices $\gamma_r$ in this representation S as follows:
Let us select a basis $\{\omega_0,\dots,\omega_{n-1}\}$ for $V^0$ such that $B^0(\omega_r, \omega_s) = \epsilon_r \delta_{rs}$.
\begin {align}
\gamma_r &= \rho(\pi(\omega_r)). \\
\gamma_r^2 &= \epsilon_r \mathbb{I}, 0 \leq r \leq n-1.\\
\gamma_r \gamma_s + \gamma_s \gamma_r &= 0, r \neq s. \\
S(l)^{-1}\gamma_r S(l) &= \sum_{s = 0}^{n - 1} a_{rs}(L^l)\gamma_s. \label {eq:spinor}
\end {align}
In the above $a_{rs}(L)$ is the r-sth matrix element of L with respect to the basis $\{\omega_0,\dots,\omega_{n-1}\}$. The Dirac matrices $\gamma_0 = \mathbb{I}, \gamma_r =  \begin{bmatrix} 0 & -\sigma_i \\ \sigma_i & 0 \end{bmatrix}, \sigma_1 =  \begin{bmatrix} 0 & 1 \\ 1 & 0 \end{bmatrix}, \gamma_i =  \begin{bmatrix} 0 & -\sigma_i \\ \sigma_i & 0 \end{bmatrix}, \sigma_2 =  \begin{bmatrix} 0 & -i \\ i & 0 \end{bmatrix}, \gamma_i =  \begin{bmatrix} 0 & -\sigma_i \\ \sigma_i & 0 \end{bmatrix}, \sigma_3 =  \begin{bmatrix} 1 & 0 \\ 0 & -1 \end{bmatrix}, r=1,2,3$.
Now, we have $G_B^0$ as the connected homogeneous Lorentz group H and $G_B^*$ is the covering group $H^* = SL(2, \mathscr{C})$.
Let $h^*$ be the Lie algebra of $H^*$ and from equation \ref {eq:spinor}, we have using the linear transformation $\delta(m), m \in SL(2, \mathscr{C})$ in $\mathscr{R}^4$ as $\eta \rightarrow m\eta m^*, \eta \in \mathscr{R}^4$ and its Lie algebra counterpart $\grave{\delta}(X)$ the following relations:
\begin {equation}
S(m)^{-1}\gamma_r S(m) = \sum_{s=0}^3 a_{rs}(\delta(m))\gamma_s, r=01,2,3.
\end {equation}
In Lie algebraic terms the above becomes
\begin {equation}
[\grave{S}(X), \gamma_r] = -\sum_{s=0}^3 a_{rs}(\grave{\delta}(X))\gamma_s , r=0,1,2,3.
\end {equation}
It can be shown that
\begin{align}
\grave{S}(X) &= \begin{bmatrix} X & 0 \\ 0 & X \end{bmatrix}. \\
S(m) &= \begin{bmatrix} m & 0 \\ 0 & (m^{-1})^* \end{bmatrix}.
\end{align}
Let us now define the spinor fields that will form the fibers of the bundles we will construct to describe the dynamics of light like particles.
\begin {defn}
Let X be the space on which the Poincar$\grave{e}$ group $\mathscr{P}^* = H^*\rtimes R^4$ acts, that is the configuration space of fundamental particles, and $\Sigma$ be the space of four complex dimensional space on which the spin representation S of $H^*$ acts. The spinor field $\psi$ is a mapping $\psi: X \rightarrow \Sigma$ such that an automorphism $(\delta(h), x)$ on X , where x is the translation and $\delta(h)$ is the homogeneous Lorentz component, the field transforms covariantly as
\begin {equation}
\psi{'}(y) = S(h)\psi((\delta(h), x)^{-1}.y), (y \in X).
\end {equation}
\end {defn}
We will encode the spinor fields into the fibers of the bundles, zero mass introduces some subtleties that are described in detail next, whose sections form the Hilbert that lead to symmetric Fock space to describe the light like fields.
\
\section {Smooth Orbits in Momentum Space}
We need to described few ingredients to construct the Hilbert space of Weyl fermions namely, the fiber bundle, the fiber vector space, an inner product for the fibers, and an invariant measure. 
The 3+1 spacetime Lorentz group $\hat{O}(3,1)$-orbits  of the momentum space $\mathscr{R}^4$, where the systems of imprimitivity established will live, described by the symmetry $\hat{O}(3,1)\rtimes\mathbb{R}^4$. The orbits have an invariant measure $\alpha^+_m$ whose existence is guaranteed as the groups and the stabilizer groups concerned are unimodular and in fact it is the Lorentz invariant measure $\frac{dp}{p_0}$ for the case of forward mass hyperboloid.  The orbits are defined as:
\begin {align}
X^{+,1/2}_m &= \{p: p^2_0 - p^2_1 - p^2_2 - p^2_3 = m^2, p_0 > 0\}, \text{\color{blue} forward mass hyperboloid}.\\
X^{+,1/2}_m &= \{p: p^2_0 - p^2_1  - p^2_2 - p^2_3 = m^2, p_0 < 0\}, \text{\color{blue} backward mass hyperboloid}. \\
X_{00} &= \{0\}, \text{\color{blue} origin}.
\end {align}
Each of these orbits are invariant with respect to $\hat{O}(3,1)$ and let us consider the stabilizer subgroup of the first orbit at p=(1,1,0,0). Now, assuming that the spin of the particle is 1/2 and mass $m \rightarrow 0$ (zero mass boson) let us define the corresponding fiber bundles (vector) for the positive mass hyperboloid that corresponds to the positive-energy states by building the total space as a product of the orbits and the group $SL(2, C)$.
\begin {align}
\hat{B}^{+,1/2}_m &= \{(p,v) \text{   }p\in{\hat{X}^{+1/2}_m, }\text{   }v\in\mathscr{C}^4,\sum_{r = 0}^3 p_r \gamma_r v = 0\}. \\
\hat{\pi} &: (p,v) \rightarrow {p}. \text{  Projection from the total space }\hat{B}^{+,2}_m \text{ to the base }\hat{X}^{+,1/2}_m.
\end {align}
It is easy to see that if $(p, v) \in B_0^{+,1/2}$ then so is also $(\delta(h)p, S(h^{*-1})v)$. Thus, we have the following Poincar$\grave{e}$ group symmetric action on the bundle that encodes spinors into the fibers:
\begin {equation}
h,(p,v) \rightarrow (p,v)^h = (\delta(h)p, S(h^{*-1})v).
\end {equation}
For $m > 0$ the fiber of $B_M^{+,1/2}$ at $p^(m) = ((a + m^2)^{1/2}, 1, 0, 0)$ is spanned by the vectors
\begin {align}
v_1^{(m)} &= \frac{1}{2}me_1 + \frac{1}{2}(1 + (1 + m^2)^{1/2})e_3. \\
v_2^{(m)} &= \frac{1}{2}me_4 + \frac{1}{2}(1 + (1 + m^2)^{1/2})e_2.
\end {align}
When we take the limit $m \rightarrow 0+$ $v_1,V-2$ converge to $e_3,e_2$ that space the fiber of $B_0^+$ at (1,1,0,0). The covering group $H^*$ is transitive on $X_m^+, X_0^+$ implies that the same convergence is true for any point. That is, if $p \in X_0^+$ then there are points $p^{(m)} \in X_m^+$ that converge to p as $m \rightarrow 0+$. This has the property that any vector v in the fiber of $B_0^+$ at p can be expressed as the limit of $v^{(m)}$ which is in fiber of $B_m^2$ at $p^{(m)}$. The same set of arguments can be applied to $B_{-m}^{+,1/2}$ implying that the bundles $B_{-m}^{+,1/2}$ also converge to $B_0^+$.
The endomorphism (chirality or helicity operator) $\Gamma = i\gamma_0\gamma_1\gamma_2\gamma_3$ transforms $B_m^{+,2}(p) \rightarrow B_{-m}^{+,1/2}(p), \forall p \in X_m^+, m > 0$ as it anticommutes with all the gammas. In the limit $\Gamma$ leaves the fibers of $B_0^+$ invariant leading to higher degenerecies with $\Gamma = \begin{bmatrix} 1 & 0 \\ 0 & -1 \end{bmatrix}$. This means $\Gamma$ commutes with all of $S(h)$ implying that $(p, v) \in B_0^+ \Rightarrow (p, \Gamma v) \in B_0^+$. If we impose either of the condition $\Gamma \psi = \pm \psi$ then we can use 2x2 the Pauli matrices for the $\gamma$s and we get the description for a Weyl fermion.

$\Gamma$ has eigen values $\pm 1$ at fiber (1, 0, 0, 1) and hence is true of all fibers. The stability group $E^*$ at (1, 0, 0, 1) given by the matrices $\begin{bmatrix} z & a \\ 0 & (z^{-1}) \end{bmatrix}, z,a \in \mathscr{C}, \abs{z}=1$, induces a representation on the fibers as $m_{z,a} \rightarrow z^{\pm 1}$.
Next step is to ensure that there is an Hermitian form on the fibers that is positive definite and left invariant with respect to S. The form $v\rightarrow p_0^{-1}\langle v, v \rangle$ can be shown to satisfy the condition and by letting $m \rightarrow 0+$ the form is still invariant and we have $E_2$ is the stabiliser group at (1,0,0,1).

Now, we can define the states of the light like particles on the Hilbert space $\hat{\mathscr{H}}^{+,1/2}_0$, square integrable functions on Borel sections of the bundle $\hat{B}^{+,1/2}_0 = \{(p,v) : (p,v) \in B_0^+, \Gamma v = \mp v\}$.

with respect to the invariant measure $\beta^{+,1}_0$, whose norm induced by the inner product is given below:
\begin {equation} \label {eq: section}
 \norm{\phi}^2 = \int_{X^+_m}p_0^{-1}\langle\phi{p},\phi{p}\rangle.{d\beta}^{+,1}_0(p).
\end {equation}
The invariant measure and the induced representation of the Poincar$\grave{e}$ group from that of the Weyl fermion are given below:
\begin {align}
{d\beta}^{+,1}_0(p) &= \frac {dp_1 dp_2 dp_3} {2(p_1^2 + p_2^2 + p_3^2)}. \\
(U_{h,x}\phi)(p) &= \text{exp i}\{x, p\} \phi (\delta(h)^{-1}p)^h.
\end {align}

\
We need Schwartz spaces, and their duals the tempered distributions
to guarantee Fourier transforms, to move with ease between configuration and momentum
space descriptions and so let us define them.
\begin {defn} Let V be an n-dimensional Euclidean space with the positive definite product $\langle .,.\rangle$ and Lebesgue measure dv. for any translation invariant differential operator D and any complex polynomial q on V, the function $\phi \rightarrow sup_{v \in V} \abs{q(v)D(\phi)(v)}$ is a seminorm on $C_c^\infty(V)$ and the collection of these seminorms induces a locally convex topology for $C_c^\infty(V)$. Its completion may be identified with $\mathscr{S}(V)$, the space of $C^\infty$ functions on V for which $sup_{v\in V} \abs{q(v)D(\phi)(v)} < \infty$ for all q and D. Intuitively, it is the function space of all infinitely differentiable functions that are rapidly decreasing at infinity along with all partial derivatives 

Let $\{v_1, \dots, v_n\}$ be an orthonormal basis for V and $D_{v_1, \dots, v_n} = (\partial/\partial x_1)^{\nu_1} \dots (\partial/\partial x_n)^{\nu_n}$ linearly span the algebra of translation invariant differential operators, and $x_1, \dots, x_n$ be the linear coordinate functions on V associated with the chosen basis. A topology on $\mathscr{S}(V)$ can be induced by the collection of seminorms \\
$sup_{x_1 \dots , x_n} abs{(1+x_1^2 + \dots + x_n^2)^k(D_{v_1, \dots, v_n}\phi)(x_1, \dots , x_n)}$ for various k=0,1,2,..., and \\
$\nu_1, \dots, \nu_n \ge 0$.  A tempered distribution E is a complex valued linear functional on $C_c^\infty(V)$ which is continuous in the topology defined. By extending these functionals to $\mathscr{S}(V)$ we may regard a tempered distribution as a continuous linear functional on the Schwartz space $\mathscr{S}(V)$.
\end {defn}
Fourier transform is an an automorphism on Schwartz space as
\begin {equation}
\hat{\phi(x)} = (2\pi)^{-n/2}\int_V exp[-i \langle x,v\rangle]\phi(v)dv, \forall \phi \in \mathscr{S}(V), x \in V.
\end {equation}
The Fourier transform of a tempered distribution $E(\phi \rightarrow E(\phi), \phi \in \mathscr{S}(V))$ is given by $\hat{E}(\phi) = E(\hat{\phi})$. We denote the dual space of $\mathscr{S}(V)$ Schwartz functions that consists of the tempered distributions as $\mathscr{S}'(V)$ and they form the Gelfand nuclear triple $\mathscr{S}(\mathscr{H}_0^{+2}) \subset \mathscr{H}_0^{+2} \subset \mathscr{S}'(\mathscr{H}_0^{+2})$.

Let us now state and establish the main result for the case of massless Weyl fermion with light-like momentum by constructing a strict cocycle from the representation of a subgroup following the prescription (lemma 5.24) in Varadarajan{'} text. The SI is a consequence of strict cocycle property and  the construction is not canonical.

lemma 5.24 \cite{Varadarajan1985}: 
Let m be a Borel homomorphism of $G_0$ into M. Then there exists a Borel map b of G into M such that 
\begin {align} \label {eq: cocycle} 
B(e) &= 1. \\
b(gh) &= b(g)m(h), \forall (g,h) \in G \times G_0.  
\end {align}
Corresponding to any such map b, there is a unique stricy (G, X)-cocycle
f such that
\begin {equation} \label {eq: cocycle2}
f(g, g^1) = b(gg^1)b(g^1)^{-1}.
\end {equation}
$\forall (g,g^{-1}) \in G \times G$. f defines m at $x_0$. Conversely, f is a strict (G, X)-cocycle and b is a Borel map such that it satisfies \ref {eq: cocycle} equation pair, then the restriction of b to $G_0$ coincides with the homomorphism m defined of at $x_0$ and b satisfies \ref {eq: cocycle2}.

\begin {thm} Light-like Weyl representation of Poincar$\grave{e}$ group is a transitive system of imprimitivity. Quantum white noise space isomorphic to the Fock space arise from a second quantized Gelfand nuclear triple. 
\end {thm}
\begin {proof}
Let $g:\mathscr{G} \rightarrow U_g(\hat{\mathscr{H}}^{+,1/2}_0)$ be a homomorphism from the two dimensional Euclidean group $\mathscr{G} = E_2$ to the unitary representation of the group in $\hat{\mathscr{H}}^{+,1/2}_0)$. We note that it is a stabilizer subgroup which is also closed at the momentum (m,0,0,m) and so $H/\mathscr{G}$ is a transitive space.

Consider a map $v(g):\mathscr{G} \rightarrow \hat{\mathscr{H}}^{+,1/2}_0$ satisfying the first order cocycle relation $v(gh) = v(g) + U_g v(h), g,h \in \mathscr{G}$.
An example of such a map is the following: \cite {KP1992} 
\begin {align*}
\mathscr{H} &= \oplus_{j=0}^\infty \mathscr{H}_j. \\
H &= 1 \oplus \oplus_{j=1}^\infty H_j. \\
U_t &= e^{-itH}, t \in \mathscr{G}. \\
v(t) &= tu_0 \oplus \oplus_{j=1}^\infty (e^{-itH_j} u_j
- u_j).
\end {align*}

Now, we can define the Weyl operator $V_g = W_g (v(g), U_g)$ where $g \in \mathscr{G}$ in the Fock space $\Gamma_s(\hat{\mathscr{H}}^{+,1/2}_0))$. 

 This is a projective unitary representation satisfying the commutator relation $V_g V_h = e^{iIm\langle v_g, U_g v_h \rangle} V_h V_g$ and let us denote the homomorphism from $\mathscr{G}$ to $V_g$ by m. This guarantees a map b that satisfies $b(gh) = b(g)m(h), g \in \mathscr{P}, h \in \mathscr{G}$ and such map can be constructed by considering the map
$c(x \rightarrow c(x)$ as Borel section of $\mathscr{P} / \mathscr{G}$ (the choice of this section not a canonical one but immaterial to our purpose here) with $c(x_0) = e$. The map
$\beta$ maps $g \in \mathscr{P} \rightarrow g \mathscr{G}$
\begin {align}
a(g) &= c(\beta (g))^{-1}. \\
b(g) &= m(a(g)). 
\end {align}
Then the strict cocycle $\phi$ satisfies $\phi(g_1, g_2) = b(g_1 g_2)b(g_2)^{-1}$.

We can now set the SI relation using the above cocycle as follows:
\begin {align*}
U_{h, x}\phi(p) &= exp i\{x, p\} \phi(g, g^{-1}x) f(g^{-1}x),  f\in\mathscr{H} \\
&\text {  character representation is defined in equation } \eqref{eq: poissonBr}.\\
P_E F &= \chi_E f \text { Position  operator}. \\
\end {align*}

We can construct the conjugate pair of field operators for the Fock space $\Gamma_s(\hat{\mathscr{H}}^{+,1/2}_0))$ as follows:

Let $p_g$ be the stone generator for the family of operators $P_{gt,p}, g \in \mathscr{P}, t \in \mathbb{R}$ and q(g) = p(ig) and we get the creation and annihilation operators as $a(g)^\dag = \frac {1}{2} ( q(g) - i p(g)$ and  $a(g) = \frac {1}{2} ( q(g) + i p(g)$.

Let us consider the Hamiltonian $H = \sigma \bullet p$, used to describe edge states in topological materials, where $\sigma$ is vector of Pauli operators and p is the momentum of the Hilbert space $\hat{\mathscr{H}}^{+,2}_0$. 
Then, the Gelfand triple we are interested in is $(E) = \Gamma_s(lim_{n\rightarrow\infty} H_n) \subset \Gamma_s(\hat{\mathscr{H}}^{+,2}_0) \subset (E^*) = \Gamma_s(lim_{n\rightarrow\infty} H_{-n})$ with respect to which the operators $a(g)^\dag$ and a(g) are white noises. Here, $H_n$ refers to n-times composition of the Hamiltonian H. These operators are defined on the white noise space (with our abuse of notation it appears to describe them on Fock space) that is isomorphic to the Fock space as $L^2((E^*), \mu) \cong \Gamma_s(\hat{\mathscr{H}}^{+,2}_0)$ and the corresponding Weyl operators have exponential functionals as domain instead of coherent vectors. 

Strictly speaking we need to build fermionic Fock space and set up the corresponding annihilation and creation operators. Instead, we have constructed the bosonic Fock space and using the unification of the \cite {Hudson1986} two spaces, $dB= (-1)^\Lambda dA$, we can construct the fermionic processes. This involves fashioning reflection processes and we leave out the details as our main focus here is the white noise process.
$\blacksquare$
\end {proof}

\section {Summary and conclusions}
We derived the covariant field operators using induced representations of groups and  expressed them in terms of systems of imprimitivity. We established the results for the massless Weyl fermion case by inducing a representation of Poincar$\grave{e}$ group from the subgroup that is a stabilizer at the momentum (1,0,0,1). Finally, we derived white noise operators on the white noise space which is isomorphic to the symmetric Fock space we constructed.



\begin{thebibliography}{00}
\bibitem {Hida1977} Hida, T.; Streit, L. On quantum theory in terms of white noise. Nagoya Math. J. 68 (1977), 21--34
\bibitem {Accardi1990} L. Accardi, A. Frigerio, and Y. Lu, The weak coupling limit as a quantum functional central limit, Commun. Math. Phys., 131 (1990), pp. 537-570.
\bibitem {Belavkin2005} J. Gough, V. P. Belavkin, and O. G. Smolyanov, Hamilton-Jacobi-Bellman equations for quantum filtering and control, J. Opt. B: Quantum Semiclass. Opt., 7 (2005), pp. 237 - 244.
\bibitem {Mackey1963} Infinite dimensional group representations, Bulletin of the American Mathematical Society, 69, 628 (1963).
\bibitem {Tezak2012} Tezak, N., Niederberger, A., Pavlichin, D. S., Sarma, G., and Mabuchi, H.: Specification of photonic circuits using quantum hardware description language. Philosophical transactions A, 370, 5270 (2012).	
\bibitem {KP1992} K. R. Parthasarathy: An Introduction to Quantum Stochastic Calculus, Birkhauser, Basel (1992)
\bibitem {Accardi2002} Luigi Accardi, Igor Volovich, and Professor Yun Gang Lu: Quantum Theory and its Stochastic Limit, Springer (2002).
\bibitem {Obata1994} N. Obata:  “White Noise Calculus and Fock Space,”  Lect.  Notes in Math.  Vol.1577, Springer, 1994.
\bibitem {Applebaum1995} D. Applebaum, Fermion stochastic calculus in Dirac-Fockspace .J.Phys.A, 28 (1995), 257-270.
\bibitem {Frigerio1989} A. Frigerio and M. Ruzzier: Relativistic transformation properties of quantum stochastic calculus, Ann. Inst. Henri Poincare 51, 67 (1989).
\bibitem {Witten2016} E. Witten, Three Lectures on Topological Phases of Matter, Nuovo Cim. Riv. Ser.39, 313 (2016).
\bibitem {Varadarajan1985} V.S. Varadarajan: Geometry of quantum theory, Springer (1985)
\bibitem {Wigner1949} Newton T D and Wigner E P 1949 Localized states for elementary systems Rev. Mod. Phys. 21 400
\bibitem {Rad2018} N. Pradeep Kumar, Radhakrishna Balu, Raymond Laflamme, C. M. Chandrashekar: Bounds on the dynamics and entanglement in a periodic quantum walks, Phys. Rev. A, (2018).
\bibitem {Rad2019} Radhakrishnan Balu: Kinematics and Dynamics of Quantum Walk in terms of Systems of Imprimitivity, J. Phys. A: Mathematical and Theoretical (2019). 
\bibitem {Rad2020} Radhakrishnan Balu: Covariant Quantum Fields via Lorentz Group Representation of Weyl Operators, Proceedings of the International Conference on Quantum Probability and Related Topics QP38 (2020).
\bibitem{vN1949} Von Neumann: On rings of operators. Reduction theory, Annals of Mathematics, 50, 401 (1949).
\bibitem {WignerLz1939} Wigner: Unitary representations of the inhomogeneous Lorentz group, Annals of mathematics, 40, 149 (1939). 
\bibitem {Kim1991} Y.S. Kim and E.S. Noz: Phase Space Picture of Quantum Mechanics: Group Theoretical Approach, Lecture Notes in Physics Series, World Scientific Pub Co Inc (1991).
\bibitem {Rad2017b} Chaobin Liu and Radhakrishnan Balu: Steady States of Continuous-time open quantum walks, Quant. Inf. 16, 173 (2017).
\bibitem {Wigner1939} Wigner: Unitary representations of the inhomogeneous Lorentz group: Ann. Math., 40, 845 (1939).
Mico Durdevic: Quantum principal bundles and corresponding gauge theories.
\bibitem {Wolf1972} J. Wolf: Spaces of constant curvature, Univ. California, Berkeley (1972).
\bibitem {Obata2007} Akihito Hora, Nobuaki Obata: Quantum Probability and Spectral Analysis of Graphs, springer (2007).
\bibitem {Schr1926} E. Schr$\ddot{o}$dinger, Naturwiss. 14 (1926) 664.
\bibitem {Hudson1986} R. L. Hudson and K. R. Parthasarathy, Unification of boson and fermion quantum stochastic calculus, Commun. Math. Phys. 104 (1986) 457-470.
\end{thebibliography}
\end{document}